\documentclass{article}

\usepackage{amsfonts}
\usepackage{amssymb,epic,eepic}
\usepackage{amsmath}
\usepackage{dcolumn}
\usepackage{bm}
\usepackage{bbm}
\usepackage{verbatim}
\usepackage{color}
\usepackage{amsthm}
\usepackage{stmaryrd}

\newcommand{\hr}{{\mathcal H}}

\newcommand{\cs}{{\mathcal S}}

\newcommand{\kr}{{\mathcal K}}

\newcommand{\cc}{{\mathbb C}}

\newcommand{\M}{{\mathcal M}}
\newcommand{\nn}{{\mathbb N}}

\newcommand{\eps}{{\varepsilon}}        

\newcommand{\cW}{\mathcal W}
\newcommand{\cX}{\mathcal X}
\newcommand{\cY}{\mathcal Y}

\newcommand{\bX}{\mathbf X}
\newcommand{\bY}{\mathbf Y}
\newcommand{\bZ}{\mathbf Z}
\newcommand{\bU}{\mathbf U}

\newcommand{\bM}{\mathbf M}
\newcommand{\bK}{\mathbf K}
\newcommand{\bL}{\mathbf L}
\newcommand{\bfn}{\mathbf n}
\newcommand{\bfm}{\mathbf m}
\newcommand{\bfx}{\mathbf x}
\newcommand{\eins}{{\mathbbm{1}}}
\newcommand{\BIGOP}[1]
{
\mathop{\mathchoice%
{\raise-0.22em\hbox{\Large $#1$}}%
{\raise-0.05em\hbox{\large $#1$}}{\hbox{\large $#1$}}{#1}}}

\newcommand{\BIGboxplus}{\mathop{\mathchoice%
{\raise-0.35em\hbox{\huge $\boxplus$}}%
{\raise-0.15em\hbox{\Large $\boxplus$}}{\hbox{\large $\boxplus$}}{\boxplus}}}

\newtheorem{theorem}{Theorem}

\newtheorem{definition}{Definition}

\newtheorem{lemma}{Lemma}

\newtheorem{remark}{Remark}

\newcommand{\tr}{\mathrm{tr}}

\begin{document}
\title{The Classical-Quantum Multiple Access Channel with Conferencing Encoders and with Common Messages}
\author{H. Boche, J. N\"otzel \\
\scriptsize{Electronic addresses: \{boche, janis.noetzel@\}tum.de}
\vspace{0.2cm}\\
{\footnotesize Lehrstuhl f\"ur Theoretische Informationstechnik, Technische Universit\"at M\"unchen,}\\
{\footnotesize 80290 M\"unchen, Germany}
}
\maketitle

\begin{abstract}
We prove coding theorems for two scenarios of cooperating encoders for the multiple access channel with two classical inputs and one quantum output. In the first scenario (ccq-MAC with common messages), the two senders each have their private messages, but would also like to transmit common messages. In the second scenario (ccq-MAC with conferencing encoders), each sender has its own set of messages, but they are allowed to use a limited amount of noiseless classical communication amongst each other prior to encoding their messages. This conferencing protocol may depend on each individual message they intend to send. The two scenarios are related to each other not only in spirit - the existence of near-optimal codes for the ccq-MAC with common messages is used for proving the existence of near-optimal codes for the ccq-MAC with conferencing encoders.
\end{abstract}
\begin{section}{Introduction}
The model of a (classical) multiple access channel has first been studied by Shannon in \cite{shannon-two-way-channels}, who formulated the model and started to analyze it. Later on Ahlswede \cite{ahlswede-multiway} and Liao \cite{liao} gave full solutions to the problem.\\
In 1983 Willems published the work \cite{willems}, introducing the model of a MAC with conferencing encoders. He gave a proof of the weak converse and a direct part that, together, built a complete coding theorem.\\
In this model, each of the encoders wants to transmit his own set of messages. But in contrast to the usual MAC model, they can both gain at least partial knowledge of the other sender's message through \emph{conferencing}: An iterative and noiseless exchange of messages under some given rate constraint. The main question then is, how the capacity region of the MAC with conferencing encoders depends on the allowed rates of the conference. Related to that, one may ask questions about the structure of an optimal conference - but at least in this model, it turns out that already a one-step conferencing protocol is enough to gain the full benefit from conferencing. Both results were obtained by Willems in \cite{willems}, who reduced the direct part to an application of the coding theorem for the MAC with a common messages that had already been solved by Slepian and Wolf in \cite{slepian-wolf-MAC}.\\
The model fits into a broader range of problems in which partial cooperation between different parties of some communication scenario is allowed. Although it seems only reasonable to assume these enhanced abilities facilitate the tasks at hand and enables a higher information throughput, it was proven recently that cooperation can also stabilize a communication system, leading to a discontinuous behaviour when switching from zero to nonzero cooperation \cite{boche-wiese}.\\
Cooperation in communication systems has generally received a lot of attention in recent years: Bross, Lapidoth, Wigger solved the case of a gaussian MAC with conferencing encoders in \cite{bross-lapidoth-wigger}, and the work \cite{wigger} by Wigger contains a broad investigation of the topic, including scenarios with feedback. The impact of conferencing has also been studied by Do, Oechtering and Skoglund for relay channels in \cite{do-oechtering-skoglund} and for the intereference channel by Marik, Yates and Kramer \cite{maric-yates-kramer}. For the case that the channel is not exactly known (compound MAC with conferencing encoders), a full coding theorem was obtained in \cite{wiese-boche-bjelakovic-jungnickel} the artbitrarily varying MAC with conferencing encoders was described in \cite{boche-wiese} who also provided a full coding theorem including the dichotomic behaviour of the deterministic capacity of that channel, and \cite{simeone-gunduz-poor-goldsmith-shamai} considered the case of an interference channel with conferencing at the decoder.
\\\\
In the present paper we extend the results of Willems to quantum channels. More precisely, we consider two senders, both of which are connected to the receiver by a ccq-MAC, a generalization of the classical setting in which the outputs of the channel are quantum states. Both senders transmit their classical messages to one receiver, who tries to decode them. A full solution of the coding problem for the ccq-MAC without conferencing has been achieved by Winter in \cite{winter}. He proved the diret part of the coding result by using timesharing, whereas the later work \cite{fawzi-hayden-savov-sen-wilde} by Fawzi, Hayden, Savov, Sen and Wilde provided, among results concerning the interference channel, a different proof of the direct part of the coding theorem for the ccq-MAC, enabling the receiver to decode both messages \emph{simultaneously}.\\
We shall use this later result together with a coding theorem for cq-channels that was developed by Winter in \cite{winter-diss} and has the property that at least partial knowledge about the codewords is given, although the codes whose existence are guaranteed by the theorem are randomly chosen. Together, these results enable us to prove the direct part of a coding theorem for the ccq-MAC with conferencing encoders. Like in the classical case, we allow the two senders to exchange messages amongst each other prior to encoding the messages that ought to be sent to the receiver. A very brief formulation of our main result then reads as follows:
\begin{center}\emph{Conferencing can enlarge the capacity region of a ccq-MAC.}\end{center}
In the classical setting, much more is known: Conferencing can for example stabilize the communication between two senders and one receiver when the channel they are transmitting over is not memoryless but arbitrarily varying. Such channel models capture for example the case where the communication line between some number of legal users is being actively manipulated by an evil party in order to prevent the communication. Good codes in such a setting are robust against a large class of clearly specified attacks, making them a good choice for applications in certain security applications. For the arbitrarily varying MAC, it turns out that already very small conferencing capacities can stabilize the whole system and boost its capacity from zero up to the maximally attainable value \cite{boche-wiese}.\\
The existence of a similar result for the quantum case seem to be a reasonable assumption, and the present paper is a first step into that direction. A second and challenging step here might be the development of coding results for the ccq compound MAC. In this model, both senders have to encode their classical messages into a set of quantum states that are being measured by the receiver - the additional assumption being, that the states are only known up to some precision, so the codes in that scenario have to work for every possibly allowed choice of the channel.\\
A standard approach from classical information theory that was developed by Ahlswede in \cite{ahlswede-elimination} is to use codes for compound channels (with one sender and one receiver, in the original setting) together with shared randomness between sender and receiver in order to obtain random codes that are robust even against arbitrarily varying noise or attacks. But this approach requires an \emph{exponential} decrease of the error for the respective compound channel model, whereas we even only have a \emph{polynomial} decrease of the error probability for the memoryless ccq-MAC. This rather slow speed of convergence comes from utilizing the coding result of \cite{fawzi-hayden-savov-sen-wilde}, so another step could be to use or develop other coding results that have the desired exponentially fast decrease of error.
\\\\
This paper is structured as follows. In Section \ref{sec:Notation}, we summarize the notation necessary in the remainder. The following Section \ref{sec:Definitions} contains the necessary definitions of codes, conferencing, achievable rates and rate regions. Section \ref{sec:Main Results} enlists our main results: A coding theorem for the ccq-MAC with conferencing encoders and another one for the ccq-MAC with a common message.\\
The rest of the paper, contained in Section \ref{sec:Proofs}, is devoted to the proofs of these results. First, in Subsection \ref{subsec:proof-of-converse-for-conferencing-MAC}, we prove the converse for the ccq-MAC with conferencing encoders.\\
Subsection \ref{subsec:direct-part-for-MAC-with-common-message} continues with the proof of the direct part for the ccq-MAC with a common message.\\
In the next subsection we apply the insights we gained in Subsection \ref{subsec:direct-part-for-MAC-with-common-message} in order to prove the direct part of the coding theorem for conferencing encoders. Since this part of our work is essentially identical to the corresponding one in \cite{willems}, we give only an outline of the basic idea in Subsection \ref{subsec:direct-part-for-MAC-with-conferencing-encoders}.\\
Finally, in Subsection \ref{subsec:converse-for-MAC-with-common-message}, we prove the converse for the ccq-MAC with a common message.
\end{section}
\begin{section}{\label{sec:Notation}Notation}
All Hilbert spaces are assumed to have finite dimension and are over the field $\cc$. The set of linear operators from $\hr$ to $\hr$ is denoted $\mathcal B(\hr)$. The adjoint of $b\in\mathcal B(\hr)$ is marked by a star and written $b^\ast$.\\
$\cs(\hr)$ is the set of states, i.e. positive semi-definite operators with trace (the trace function on $\mathbb B(\hr)$ is written $\tr$) $1$ acting on the Hilbert space $\hr$. Pure states are given by projections onto one-dimensional subspaces. A vector $x\in\hr$ of unit length spanning such a subspace will therefore be referred to as a state vector, the corresponding state will be written
$|x\rangle\langle x|$. For a finite set $\mathbf X$ the notation $\mathfrak{P}(\mathbf X)$ is reserved for the set of probability distributions on $\mathbf X$, and
$|\mathbf X|$ denotes its cardinality. For any $l\in\nn$, we define $\bX^l:=\{(x_1,\ldots,x_l):x_i\in\bX\ \forall i\in\{1,\ldots,l\}\}$, we also write $x^l$ for the elements of $\bX^l$. Associated to every such element is a function $N(\cdot|x^l):\bX\to\nn$ defined by $N(x|x^l):=|\{i:x_i=x\}|$.\\
The set of classical-quantum channels (abbreviated here using the term 'cq-channels') with finite input alphabet $\mathbf Z$ and output system $\kr$ is denoted $CQ(\mathbf Z,\kr)$.\\
For any natural number $N$, we define $[N]$ to be the shortcut for the set $\{1,...,N\}$.\\
Using the usual operator ordering symbols $\leq$ and $\geq$ on $\mathcal B(\hr)$, the set of measurements with $N\in\nn$ different outcomes is written
\begin{align}\M_N(\hr):=\{\mathbf D:\mathbf D=(D_1,\ldots,D_N)\ \wedge\ \sum_{i=1}^ND_i\leq\eins_\hr\ \wedge\ D_i\geq0\ \forall i\in[N]\}.\end{align}
To every $\mathbf D\in \M_N(\hr)$ there corresponds a unique operator defined by $D_0:=\eins_\hr-\sum_{i=1}^ND_i$. Throughout the paper, we will assume that $D_0=0$ holds. This is possible in our scenario, since adding the element $D_0$ to any of the other $D_1,\ldots,D_N$ does not decrease the performance of a given code.\\
The von Neumann entropy of a state $\rho\in\mathcal{S}(\hr)$ is given by
\begin{equation}S(\rho):=-\textrm{tr}(\rho \log\rho),\end{equation}
where $\log(\cdot)$ denotes the base two logarithm which is used throughout the paper.\\
The Holevo information is for a given channel $\mathcal W \in CQ(\mathbf{Z},\hr)$ and input probability distribution $p \in \mathfrak P(\mathbf{X})$ defined by
\begin{align}
 \chi(p, \mathcal W) := S(\overline{\mathcal W}) - \sum_{z \in \mathbf{Z}} p(z) S(\mathcal W(z)),
\end{align}
where $\overline{\mathcal W}$ is defined by $\overline{\mathcal W} := \sum_{z \in \mathbf{Z}} p(z) \mathcal W(z)$. We shall employ a slightly different notation that is closer to the one used in the classical scenario. To the distribution $p$ we can always associate a random variable $Z$ with values in $\bZ$ that is distributed according to $p$. If we label the physical system that is modelled on the Hilbert space $\kr$ by $Q$, we can define
\begin{align}
\chi(Z;Q):=\chi(p,\mathcal W).
\end{align}
If our channel has a bipartite input ($\bZ=\bX\times\bY$), and $(X,Y)$ is a random variable on $\bX\times\bY$ that is distributed according to $\mathbb P((X,Y)=(x,y))=p(y)q(x|y)$ it even makes sense to define the quantity
\begin{align}
\chi(X;Q|Y):=\sum_{y\in\bY}p(y)\chi(q(\cdot|y),\mathcal W(\cdot\times y)).
\end{align}
Whenever necessary, the elements $x$ of some finite set $\bX$ will be identified with a set $\{|\bfx\rangle\langle \bfx|\}_{x\in\bX}\subset \mathbb \mathcal B(C^{|\bX|})$ of matrix units that are pairwise orthogonal (with respect to the Hilbert Schmidt inner product).
\end{section}
\begin{section}{\label{sec:Definitions}Definitions}
In the remainder, $\mathcal W\in\mathcal C(\bX\times\bY,\kr)$ will denote a classical, classical - quantum multiple access channel (ccq-MAC). The quantum part of the system will also be referred to by the symbol $Q$ and, given a probability distribution on the input system of the channel, the corresponding random variable will be written $(X,Y)$. Further random variables may arise.
\begin{definition}[Codes for the ccq-MAC with conferencing encoders]\label{def:codes-MAC-conferencing}
For given $l\in\nn$, an $(M_l,N_l,C,D)$ code $\mathfrak C_l)$ for the ccq-MAC with encoders conferencing at rates $C\geq0$ and $D\geq0$ consists of:
\begin{enumerate}
\item Two natural numbers $M_l$ and $N_l$ that form the message sets $[M_l]$ and $[N_l]$.
\item Positive numbers $C,D$ that give upper bounds on the overall rate of a \emph{conference}. This conference consists of: a natural number $K\in\nn$, finite message sets $V_{l,1},\ldots,V_{l,K}$ and $W_{l,1},\ldots,W_{l,K}$ ($V_{l,0}=W_{l,0}=\emptyset$ in order to have more compact notation) and conferencing functions
    \begin{align}
    g_{l,i}:[M_l]\times (\times_{j=0}^{i-1}V_{l,j})\times(\times_{j=0}^{i-1}W_{l,j})\mapsto W_{l,i},\qquad i\in[K],\\
     f_{l,i}:[M_l]\times (\times_{j=0}^{i-1}W_{l,j})\times(\times_{j=0}^{i-1}V_{l,j})\mapsto V_{l,i},\qquad i\in[K]
     \end{align}
     such that $\sum_{k=1}^K\log|V_{l,k}|\leq C$ and $\sum_{k=1}^K\log|W_{l,k}|\leq D$.\\
     The outcomes of the conference are stored in the set $U_l:=\prod_{i=1}^KW_i\times\prod_{i=1}^KV_i$. If the codewords $(n,m)$ were sent, they are given by arrays that will be written
     \begin{align}
     \mathcal C_l(m,n)&=(m,g_1(n),g_2(n,f_1(m)),g_3(n,f_1(m),f_2(m,g_1(n))),\ldots)\\
     \mathcal D_l(m,n)&=(n,f_1(m),f_2(m,g_1(n)),f_3(m,g_1(n),g_2(n,f_1(m))),\ldots).
     \end{align}
\item Two functions $f_l$ and $g_l$ such that $f_l$ takes as inputs the outcomes $\mathcal C_l(m,n)$ and $g_l$ the outcomes $\mathcal D_l(m,n)$ of the conference and $f_l$ outputs a corresponding codeword in $\bX^l$, while $g_l$ gives one in $\bY^l$.
\item A POVM $\mathbf D^l=\{D^l_{mn}\}_{m,n=1}^{M_l,N_l}\in\mathcal M_{M_l\cdot N_l}$.
\item Denoting the code by the corresponding outcomes $\mathcal C_l,\mathcal D_l$ of the conference, we can write its average success probability as
\begin{align}
p_{\mathrm{s}}(\mathfrak C_l)=\frac{1}{M_l}\frac{1}{N_l}\sum_{m=1}^{M_l}\sum_{n=1}^{N_l}\tr\{D^l_{mn}\mathcal W^{\otimes l}(f_l(\mathcal C_l(m,n))\times g_l(\mathcal D_l(m,n)))\}.
\end{align}
\end{enumerate}
\end{definition}
\begin{definition}[Achievability for the ccq-MAC with conferencing encoders]
A pair $(R_M,R_N)$ of nonnegative real numbers is said to be achievable for the ccq-MAC with encoders conferencing at rates $C\geq0$ and $D\geq0$ if there is a sequence $(\mathfrak C_l)_{l\in\nn}$ of codes as in Definition \ref{def:codes-MAC-conferencing} with conferencing rates $C$ and $D$ such that
\begin{align}
\liminf_{l\to\infty}\frac{1}{l}\log M_l\geq R_M,\ \ \liminf_{l\to\infty}\frac{1}{l}\log N_l\geq R_N\ \ \mathrm{and}\ \ \liminf_{l\to\infty}p_{\mathrm{s}}(\mathfrak C_l)=1.
\end{align}
\end{definition}
\begin{definition}[Capacity region of the ccq-MAC with conferencing encoders]
The capacity region $C(\mathcal W,C,D)$ of the ccq-MAC with encoders conferencing at rates $C\geq0$ and $D\geq0$ is defined to be the closure of the set of all rates that are achievable (for the ccq-MAC, with conferencing at rates $C$ and $D$).
\end{definition}
\begin{definition}[Codes for the ccq-MAC with common messages]\label{def:codes-MAC-common}
For $l\in\nn$, a code $\mathfrak C_l$ for the ccq-MAC with common messages consists of a triple $(K_l,T_l,M_l)$ of natural numbers, two encoding functions $f_l:[K_l]\times[M_l]\to\bX^l$, $g_l:[T_l]\times[M_l]\to\bY^l$, and a POVM $(\Lambda_{k,t,m})_{k,l,m=1}^{K_l,T_l,M_l}$. The success probability of the code is given by
\begin{align}
p_{\mathrm{s}}(\mathfrak C_l):=\frac{1}{K_l\cdot T_l\cdot M_l}\sum_{k,t,l=1}^{K_l,T_l,M_l}\tr\{\Lambda_{k,t,l}\mathcal W^{\otimes l}(f_l(k,m)\times g_l(t,m))\}.
\end{align}
\end{definition}
\begin{definition}[Achievability for the ccq-MAC with common messages]
A triple $(S_X,S_Y,S_C)$ of nonnegative real numbers is said to be achievable for the ccq-MAC with common messages if there exists a sequence of $(\mathfrak C_l)_{l\in\nn}$ of codes as in Definition \ref{def:codes-MAC-common} such that
\begin{align}
\liminf_{l\to\infty}\frac{1}{l}\log K_l\geq S_X,\ \ &\liminf_{l\to\infty}\frac{1}{l}\log T_l\geq S_Y,\ \ \liminf_{l\to\infty}\frac{1}{l}\log M_l\geq S_C\\ &\mathrm{and}\ \ \liminf_{l\to\infty}p_{\mathrm{s}}(\mathfrak C_l)=1.
\end{align}
\end{definition}
\begin{definition}[Capacity region of the ccq-MAC with common messages]
The capacity region of the ccq-MAC $\mathcal W$ with common messages is given by the closure of the set of all rate triples that are achievable (for $\mathcal W$, with common message).
\end{definition}
\end{section}
\begin{section}{\label{sec:Main Results}Main Results}
Our main results are two complete coding theorems: One for the ccq-MAC with conferencing encoders, the other for the ccq-MAC with a common message. This joint presentation is not just by chance: The direct part of the coding theorem for the model with a joint message serves as a building block for the model with conferencing senders.\\
We now state our theorems, in the same order as their proofs are given later. The first one is an outer bound on the capacity region of a ccq-MAC with conferencing encoders:
\begin{theorem}[Converse of the coding theorem for ccq-MAC with conferencing encoders]\label{theorem:converse-for-conferencing-MAC}
For the ccq-MAC with conferencing encoders, a rate pair $(R_X,R_Y)$ is achievable only if it is contained in the set
\begin{align}\mathfrak R_{\mathrm{conf}}(\mathcal W,C,D):=\mathrm{cl}(\cup_p\mathfrak R_{p,\mathrm{conf}}(\mathcal W,C,D))\end{align} defined by the sets $\mathfrak R_{p,\mathrm{conf}}(\mathcal W,C,D)$ of all pairs of real nonnegative numbers $(R_N,R_M)$ satisfying
\begin{align}
R_M&\leq \chi(X;Q|Y,U)+C\\
R_N&\leq \chi(Y;Q|X,U)+D\\
R_M+R_N&\leq \chi(X,Y;Q|U)+C+D\\
R_M+R_N&\leq \chi(X,Y;Q)
\end{align}
where the states used to evaluate the entropic quantities on the right hand sides are defined by
\begin{align}
\sum_{u,x,y}p(u,x,y)|u\rangle\langle u|\otimes|x\rangle\langle x|\otimes|y\rangle\langle y|\otimes\mathcal W(x,y)
\end{align}
and the distribution $p\in\mathfrak P(\bU\times\bX\times\bY)$ can be decomposed such that $p(u,x,y)=q(u)r(x|u)s(y|u)$ for suitable distributions $q\in\mathfrak P(\bU)$, where $s(\cdot|u)\in\mathfrak P(\bX)$ and $r(\cdot|u)\in\mathfrak P(\bY)$ for every $u\in\bU$. Finally, the cardinality of the alphabet $\bU$ can be restricted by the cardinality bound $|\bU|\leq|\cX|\cdot|\cY|+3$.
\end{theorem}
Second, we prove the existence of codes that transmit common messages as well as individual messages of two senders over a ccq-MAC with asymptotically vanishing average error probability, at certain rates. This means that we can give an inner bound on the capacity region of that model. The result is used afterwards to obtain a direct coding theorem for the ccq-MAC with conferencing encoders as well.
\begin{theorem}[Direct part of coding theorem for the ccq-MAC with a common message]\label{theorem:direct-part-of-MAC-with-common-message}
Every rate triple $(R_C,R_X,R_Y)$ satisfying $(R_C,R_X,R_Y)\in\mathfrak R_{\mathrm{comm}}(\mathcal W)$ is achievable. The convex set $\mathfrak R_{\mathrm{comm}}(\mathcal W)$ is given by \begin{align}\mathfrak R_{\mathrm{comm}}(\mathcal W)=\mathrm{cl}(\cup_q\mathfrak R_{q,\mathrm{comm}}(\mathcal W)),\end{align} where the sets $\mathfrak R_{q,\mathrm{comm}}(\mathcal W)$ are given by all triples $(S_C,S_X,S_Y)$ satisfying below inequalities for a distribution $q\in\mathfrak\bU\times\bX\times\bY$ having the structure $q(x,y,u)=p(u)r(x|u)s(y|u)\ \forall (u,x,y)\in\bU\times\bX\times\bY$ and with the overall cq state being $\sum_{u,x,y}q(u,x,y)|u\rangle\langle u|\otimes|x\rangle\langle x|\otimes|y\rangle\langle y|\otimes\mathcal W(x,y)$.
\begin{align}
S_X&\leq \chi(X;Q|Y,U)\\
S_Y&\leq \chi(Y;Q|X,U)\\
S_X+S_Y&\leq \chi(X,Y;Q|U)\\
S_C+S_X+S_Y&\leq \chi(X,Y;Q)
\end{align}
\end{theorem}
As was the case in the classical paper \cite{willems} by Willems, the existence of a coding result for the ccq-MAC with private and common messages enables one to prove the direct part of the coding theorem for the ccq-MAC with conferencing encoders, leading to the following result:
\begin{theorem}[Direct part of coding theorem for ccq-MAC with conferencing encoders]\label{theorem:direct-part-for-MAC-with-conferencing-encoders}
Every rate pair $(R_X,R_Y)\in\mathfrak R_{\mathrm{conf}}(\mathcal W)$ is achievable, thus $C(\mathcal W,C,D)=\mathfrak R_{\mathrm{conf}}(\mathcal W,C,D)$.
\end{theorem}
At last, and in order to have a coherent and self-contained presentation, we also prove the converse theorem for the ccq-MAC with common and private messages. This part, as well as the direct part for the ccq-MAC with conferencing encoders, shows that the two models are in fact closely related from an information theoretic point of view.
\begin{theorem}[Converse for the ccq-MAC with common message]\label{theorem:converse-for-MAC-with-common-message}
For the ccq-MAC with common message, no rate triple $(R_C,R_X,R_Y)$ outside of $\mathfrak R_{\mathrm{comm}}(\mathcal W)$ is achievable.
\end{theorem}
\begin{remark}
Above results, put together, establish the region $\mathfrak R_{\mathrm{conf}}(\mathcal W,C,D)$ as the rate region of the ccq-MAC $\mathcal W$with senders conferencing at rates $C,D$ and the region $\mathfrak R_{\mathrm{comm}}(\mathcal W)$ as the rate region for the same model but with common messages instead of conferencing senders.
\end{remark}
\end{section}
\begin{section}{\label{sec:Proofs}Proofs}
This section contains the proofs to all our four statements, first the converse for the ccq-MAC with conferencing encoders, then the direct part of the coding theorem for the ccq-MAC with a common message. The latter one directly leads to a proof of the direct part of the coding theorem for the case of conferencing encoders. Finally, we provide a proof of the converse for the ccq-MAC with common message.
\begin{subsection}{\label{subsec:proof-of-converse-for-conferencing-MAC}Proof of the converse for the ccq-MAC with conferencing encoders}
\begin{proof}[Proof of Theorem \ref{theorem:converse-for-conferencing-MAC}]
Let us first derive the bound on the cardinality of the set $\bU$. To this end, consider Lemma 3 in \cite{ahlswede-koerner}. For states of the form used in Theorem \ref{theorem:converse-for-conferencing-MAC}, it ensures that for every such state defined using an arbitrary set $\bU'$, distribution $q'\in\mathfrak P(U')$ and corresponding distributions $r(\cdot|u),\ s(\cdot|u)$ there is another set $\bU$ and a $q\in\mathfrak P(U)$ such that the cardinality of the set $\bU$ is bounded by $|\bU|\leq|\cX|\cdot|\cY|+3$ and
\begin{align}
\chi(X';Q'|Y',U')&=\chi(X;Q|Y,U)\\
\chi(Y';Q'|X',U')&=\chi(Y;Q|X,U)\\
\chi(X',Y';Q'|U')&= \chi(X,Y;Q|U)\\
\chi(X',Y';Q')&=\chi(X,Y;Q),
\end{align}
where the primes indicate that the overall system depends on our choice of the variable $U$ distributed according to $q$ or $U'$ distributed according to $q'$. It is worth noting that Lemma 3 in \cite{ahlswede-koerner} is applied such that even $(X',Y')=(X,Y)$, and the fourth of the above four equalities is exactly due to this fact.\\
An analoguous reasoning applies to the case of the ccq-MAC with common messages.
\\\\
We will now prove the converse theorem. Given an $(M_l,N_l,C,D)$ code $\mathfrak C_l$, we can define the states
$\sigma_{mn}:=|\bfm\rangle\langle \bfm|\otimes|\bfn\rangle\langle\bfn|\in\cs(\mathbb C^{M_l}\otimes\mathbb C^{N_l})$ and $\sigma'_{mn}:=|\mathbf{\mathcal C_l(m,n)}\rangle\langle\mathbf{\mathcal C_l(m,n)}|\otimes|\mathbf{\mathcal D_l(m,n)}\rangle\langle\mathbf{\mathcal D_l(m,n)}|\in\cs(\mathbb C^{C+D}\otimes\mathbb C^{C+D})$ for some arbitrary but fixed orthonormal bases in the respective spaces, and for every message pair $(m,n)\in[M_l]\times[N_l]$. They define the overall state
\begin{align}
\rho_l:=\frac{1}{M_l\cdot N_l}\sum_{m,n=1}^{M_l,N_l}\sigma_{mn}\otimes\sigma'_{mn}\otimes\cW^{\otimes l}((f_l\circ\mathcal C_l)\times(g_l\circ\mathcal D_l)(m,n)),
\end{align}
that we will be using in order to calculate our entropic quantities whenever the specific measurement outcomes of the decoder are of no importance. Since we intend to apply the Holevo bound within the first steps of our proof, this will soon be the case. The state $\rho_l$ contains all the information that is contained in the process of randomly selecting the messages, then evaluate the outcome of the conference for the specific messages, an encoding that depends on both the messages and the outcome of the conference, and producing a quantum state that can be used for decoding at the receiver's side.\\
Using the POVM $\mathbf D^l$ as a completely positive map $\mathbf D^l\in\mathcal C(Q^{\otimes l},Q^{\otimes l}\otimes\hr'_{M_l}\otimes\hr'_{N_l})$ (where $\hr'_{M_l}=\mathbb C^{M_l}$ and $\hr'_{N_l}=\mathbb C^{N_l}$) defined via
\begin{align}
\mathbf D^l(A):=\sum_{m,n=1}^{M_l,N_l}\tr\{AD^l_{m,n}\}|\bfm\rangle\langle\bfm|\otimes|\bfn\rangle\langle\bfn|\qquad\forall\ A\in\mathcal B(\kr^{\otimes l})
\end{align}
we further define the state after the sender has applied its recovery operation as
\begin{align}
\sigma_l:=(Id_{\hr_{M_l}\otimes\hr_{N_l}\otimes \hr_{U_l}}\otimes\mathbf D^l)(\rho_l).
\end{align}
To get started, let there be a sequence $(\mathfrak C_l)_{l\in\nn}$ of codes for the MAC $\mathcal W$ with encoders conferencing at rates $C$ and $D$. Let
\begin{align}
p_{\mathrm{s}}(\mathfrak C_l)\geq1-\eps_l
\end{align}
for some sequence $(\eps_l)_{l\in\nn}$ satisfying $\searrow\eps_l=0$. Fix $l$, for the time being. We will let it tend to infinity at the end of the proof. An application of Fano's inequality together with our prerequisites yields that
\begin{align}
H(M^l,N^l|\hat M^l,\hat N^l)\leq\eps_l\cdot\log(M_l\cdot N_l)+1
\end{align}
holds for the common distribution of the messages $(M^l,N^l)$ at senders and $(\hat{M}^l,\hat{N}^l)$ at receivers side. Note that, while $M_l$ and $N_l$ denote the numbers of messages, the symbols $M^l$ and $N^l$ denote random variables.\\
Let us now introduce the abbreviation $\delta_l:=\eps_l\cdot\log(M_l\cdot N_l)+1$. It follows
\begin{align}
\log(M_l)&=H(M^l|N^l)\leq I(M^l;\hat{M}^l,\hat{N}^l,U^l|N^l)+\delta_l\\
\log(N_l)&=H(N^l|M^l)\leq I(N^l;\hat{M}^l,\hat{N}^l,U^l|M^l)+\delta_l\\
\log(N_l\cdot M_l)&=H(N^l,M^l)\leq I(N^l,M^l;\hat{N}^l,\hat{M}^l,U^l)+\delta_l.
\end{align}
Above terms can trivially, by definition of the conditional entropic quantities involved, be split up as follows:
\begin{align}
I(M^l;\hat{M}^l,\hat{N}^l,U^l|N^l)&=I(M^l;U^l|N^l)+I(M^l;\hat{M}^l,\hat{N}^l|N^l,U^l)\\
I(N^l;\hat{M}^l,\hat{N}^l,U^l|M^l)&=I(N^l;U^l|M^l)+I(N^l;\hat{M}^l,\hat{N}^l|M^l,U^l)\\
I(N^l,M^l;\hat{M}^l,\hat{N}^l,U^l)&=I(N^l,M^l;U^l)+I(N^l,M^l;\hat{M}^l,\hat{N}^l|U^l).
\end{align}
From the results of \cite{willems}, his inequalities $(11),\ (12),\ (13)$ we know that
\begin{align}
I(M^l;U^l|N^l)&\leq l\cdot C\\
I(N^l;U^l|M^l)&\leq l\cdot D\\
I(N^l,M^l;U^l)&\leq l\cdot(C+D).
\end{align}
Now consider the following three conditional independence relations that were proven in \cite{willems}:\\
First (equation (20) in \cite{willems}), the random variable $(M^l,N^l,U^l)$ is distributed as
\begin{align}
\mathbb P((M^l,N^l,U^l)=(m,n,u))=\mathbb P(M^l=m|U^l=u)\cdot\mathbb P(N^l=n|U^l=u)\cdot\mathbb P(U^l=u).\label{eqn:willems-20}
\end{align}
Second (equation (21) in \cite{willems}), and as a consequence of above equalities, the distribution of the random variable $(X^l,Y^l,U^l)$ obeys
\begin{align}
\mathbb P((X^l,Y^l,U^l)=(x^l,y^l,u))=\mathbb P(X^l=m|U^l=u)\cdot\mathbb P(Y^l=n|U^l=u)\cdot\mathbb P(U^l=u).\label{eqn:willems-21}
\end{align}
Third, this implies that also for each $i\in[l]$ the random variable $(U^l,X_i,Y_i)$ is distributed according to
\begin{align}
\mathbb P((U^l,X_i,Y_i)=(u^l,x,y))=\mathbb P(X_i=x|U^l=u^l)\cdot\mathbb P(Y_i=y|U^l=u^l)\cdot\mathbb P(U^l=u^l).\label{eqn:willems-22}
\end{align}
In order to have no confusion arising about the nature of this conditional independence let us note that in all three of the above statements the independence is conditioned on the outcome of the whole conference, which can in general not be single-letterized!\\
We will also need the following subadditivity relation for the Holevo quantity, that we borrow from \cite{winter-diss}:
\begin{lemma}\label{lemma:winter-subadditivity}
For CQ channels $\mathcal W_1,\mathcal W_2\in CQ(\bX,\kr)$ (their in- and output systems will be denoted $X_i$ and $Q_i$, $i=1,2$) and a probability distribution $p\in\mathfrak P(\bX\times\bX)$ we have, for the overall state
\begin{align}
\rho_{X_1X_2,Q_1,Q_2}:=\sum_{x_1,x_2}p(x_1,x_2)|x_1\rangle\langle x_1|\otimes|x_2\rangle\langle x_2|\otimes\mathcal W_1(x_1)\otimes\mathcal W_2(x_2),
\end{align}
\begin{align}
\chi(X_1,X_2;Q_1,Q_2)\leq \chi(X_1;Q_1)+I(X_2;Q_2).
\end{align}
\end{lemma}
\begin{remark}\label{remark:winter-subadditivity}
It is clear that above theorem holds as well if there is a conditioning on a third classical system, e.g. $\chi(X_1,X_2;Q_1,Q_2|U)\leq \chi(X_1;Q_1|U)+\chi(X_2;Q_2|U)$.
\end{remark}
We now continue with the proof of Theorem \ref{theorem:converse-for-conferencing-MAC}. Together with equation (\ref{eqn:willems-20}) and the Holevo bound, above Lemma enables us to validate the following chain of inequalities:
\begin{align}
I(M^l;\hat{M}^l,\hat{N}^l|N^l,U^l)&\leq I(X^l;\hat{M}^l,\hat{N}^l|Y^l,U^l)\\
&\leq \chi(X^l;Q^l|Y^l,U^l)\\
&\leq\sum_{i=1}^l\chi(X_i;Q_i|Y^l,U^l).
\end{align}
But both $X_i$ and $Q_i$ are independent of $Y_1,\ldots,Y_{i-1},Y_{i+1},\ldots,Y_{l}$ given $U^l=u$, thus for every $i\in[l]$ we get
\begin{align}
\chi(X_i;Q_i|Y^l,U^l)=\chi(X_i;Q_i|Y_i,U^l),
\end{align}
establishing
\begin{align}
I(M^l;\hat{M}^l,\hat{N}^l|N^l,U^l)&\leq\sum_{i=1}^l\chi(X_i;Q_i|Y_i,U^l).
\end{align}
In the very same manner we can prove that
\begin{align}
I(N^l;\hat{M}^l,\hat{N}^l|M^l;U^l)&\leq\sum_{i=1}^l\chi(Y_i;Q_i|X_i,U^l).
\end{align}
Using the Holevo bound and Lemma \ref{lemma:winter-subadditivity} we can also easily establish
\begin{align}
I(N^l,M^l;\hat{N}^l,\hat{M}^l|U^l)&\leq\sum_{i=1}^l\chi(X_i,Y_i;Q_i|U^l).
\end{align}
A further consequence of the Holevo bound and Lemma \ref{lemma:winter-subadditivity} is the upper bound
\begin{align}
I(M^l,N^l;\hat{N}^l,\hat{M}^l)&\leq\sum_{i=1}^l\chi(X_i,Y_i;Q_i).
\end{align}
Putting together what we found out so far we see that the following inequalities hold:
\begin{align}
\log(M_l)&\leq\sum_{i=1}^l\chi(X_i;Q_i|Y_i,U^l)+l\cdot C+\delta_l,\label{eqn-31}\\
\log(N_l)&\leq\sum_{i=1}^l\chi(Y_i;Q_i|Y_i,U^l)+l\cdot D+\delta_l,\label{eqn-32}\\
\log(M_l\cdot N_l)&\leq\sum_{i=1}^l\chi(X_i,Y_i;Q_i|U^l)+l\cdot(C+D)+\delta_l,\label{eqn-33}\\
\log(M_l\cdot N_l)&\leq\sum_{i=1}^l\chi(X_i,Y_i;Q_i)+\delta_l\label{eqn-34}.
\end{align}
We are going to show that this implies the following: Every pair of achievable rates is arbitrarily close to a convex combination of points taken from the rate region described in Theorem \ref{theorem:converse-for-conferencing-MAC}.\\
But that will imply that $(R_M,R_N)\in\mathfrak R_{\mathrm{conf}}(\mathcal W)$, since the latter is a closed set.\\
Let us abbreviate the above regularized sums over mutual informations as follows:
\begin{align*}
P_{1,l}:=\sum_{i=1}^l\chi(X_i;Q_i|Y_i,U^l), \ P_{2,l}:=\sum_{i=1}^l\chi(Y_i;Q_i|Y_i,U^l),\\
P_{3,l}:=\sum_{i=1}^l\chi(X_i,Y_i;Q_i|U^l), \ P_{4,l}:=\sum_{i=1}^l\chi(X_i,Y_i;Q_i).
\end{align*}
It is then evident that the following equalities hold:
\begin{align}
P_{1,l}=\sum_{j=1}^l\frac{1}{l}\chi_{p_j}(X;Q|Y,U),\qquad P_{2,l}=\sum_{j=1}^l\frac{1}{l}\chi_{p_j}(Y;Q|X,U),\\
P_{3,l}=\sum_{j=1}^l\frac{1}{l}\chi_{p_j}(X,Y;Q|U),\qquad P_{4,l}=\sum_{j=1}^l\frac{1}{l}\chi_{p_j}(X,Y;Q),
\end{align}
where the Holevo quantities are evaluated over states as described in Theorem \ref{theorem:converse-for-conferencing-MAC} with respective probability distributions $p_j$ defined by $p_j(u,x,y):=\mathbb P((U^l,X_j,Y_j)=u,x,y)$. As can be seen from equation (\ref{eqn:willems-22}), these distributions obey the structure that is stated in Theorem \ref{theorem:converse-for-conferencing-MAC}. Going further in our discussion we see that the points $P_{1,l},\ldots,P_{4,l}$ are contained in the four dimensional closed (and also convex, due to the freedom in the choice of the alphabets $\bU$) region $\mathfrak R^4$ defined by
\begin{align*}
\mathfrak R^4:=\left\{\right.(R_i)_{i=1}^4|&R_1\leq \chi(X;Q|Y,U)+C,\ R_2\leq \chi(Y;Q|X,U)+D,\ R_3\leq \chi(X,Y;Q|U)+C+D,\\
&R_4\leq \chi(X,Y;Q)\ \mathrm{for\ some\ alphabet\ }\bU\ \mathrm{and\ }(U,X,Y,Q)\ \mathrm{a\ classical-quantum}\\
&\mathrm{system\ as\ in\ Theorem\ \ref{theorem:converse-for-conferencing-MAC}}\left\}\right.
\end{align*}
for all $l\in L$. But then for all $\eps>0$ and $l\geq L=L(\eps)$ we have that
\begin{align}
R_M\leq P_{1,l}+3\eps,\qquad R_N\leq P_{2,l}+3\eps\\
R_M+R_N\leq\min\left\{P_{3,l}+3\eps,\ P_{4,l}+3\eps\right\}
\end{align}
and whence for all $\eps>0$ the vector $(R_M,R_N,R_M+R_N,R_M+R_N)$ is contained in $\mathfrak R^4+B_{3\eps}$, where for two sets $A,B\subset \mathbb R^4$ their sum is defined by $A+B:=\{a+b:a\in A,\ b\in B\}$ and for an $\eps>0$ we denote a small ball of radius $\eps$ (in one-norm) around the point $0\in\mathbb R^4$ by $B_\eps:=\{x\in\mathbb R^4:\sum_i|x_i|\leq\eps\}$.\\
Therefore $(R_M,R_N,R_M+R_N,R_M+R_N)\in\mathrm{cl}(\mathfrak R^4)$, and from that we deduce that $(R_M,R_N)\in\mathrm{cl}(\cup_p\mathfrak R_{p,\mathrm{conf}}(\mathcal W))=\mathfrak R_{\mathrm{conf}}(\mathcal W)$, which completes the proof of the converse for the MAC with conferencing encoders.
\end{proof}
\end{subsection}
\begin{subsection}{\label{subsec:direct-part-for-MAC-with-common-message}Construction of codes for the MAC with a common message}
In our proof of Theorem \ref{theorem:direct-part-of-MAC-with-common-message}, we will need the following theorem which is a suitably reformulated version of Theorem 10 in \cite{winter}:
\begin{theorem}\label{thm:winter-existence-of-codes}
For $\lambda,\tau\in(0,1)$ there is a number $K(\lambda,\tau,|\bU|,d)$ such that for every cq channel $\mathcal T\in CQ(\bU,\kr)$, probability distribution $q\in\mathfrak P(\bU)$, $l\in\nn$, and $\mathcal A\subset\bU^l$ with $q^{\otimes l}(\mathcal A)\geq\tau$ there are codewords $\{u^l_m\}_{m=1}^{M_n}$ and a decoding POVM $\{D_m\}_{m\in M_l}$ on $\kr^{\otimes l}$ with the properties
\begin{align}
\forall m\in[M_l],\ u^l_m\in\mathcal A\ \ \mathrm{and}\ \ M_l\geq2^{l(\chi(q,\mathcal T)-\frac{1}{\sqrt{l}}K(\lambda,\tau,|\bU|,d))}\ \ \mathrm{and}\ \ \min_{m\in[M_l]}\tr\{\mathcal T^{\otimes l}(u^l_m)D_m\}\geq1-\lambda.
\end{align}
It holds $K(\lambda,\tau,|\bU|,d)=K'd|\bU|\sqrt{2d/\lambda}+K'|\bU|\sqrt{2|\bU|d/ \tau}\log(d)$ for some universal constant $K'$.
\end{theorem}
The benefits of codes constructed from application of this theorem are, that some control over the structure of the codewords is given. We will now begin with the main topic of this section:
\begin{proof}[Proof of Theorem \ref{theorem:direct-part-of-MAC-with-common-message}]
Take any finite set $\bU$. Take any distribution $p\in\mathfrak P(\bU)$, and for each $u\in\bU$ take two (conditional) probability distributions $r(\cdot|u)\in\mathfrak P(\bX)$ and $s(\cdot|u)\in\mathfrak P(\bY)$. We may then define a cq-channel $\mathcal V\in CQ(\bU,\kr)$ by
\begin{align}
\mathcal V(u):=\sum_{x\in\bX}\sum_{y\in\bY}r(x|u)s(y|u)\mathcal W(x,y).
\end{align}
The construction of the code will be in two steps. First we find a code for the common message, then (conditioned on that code) we construct a code for the private messages. Intuitively speaking, the receiver tries to decode the common messages first, and afterwards applies a decoder for the private messages.\\
\textbf{Construction of code for the common message:}
We choose $\lambda$ dependent on $l$, more precisely we set $\lambda_l:=l^{-1/4}$ and apply this theorem to the channel $\mathcal V$. The distribution on its input system will be $p$, and the set $\mathcal A$ will be dependent on $l$ as well: we choose sets $\mathcal A_l:= T_{p,l^{-1/8}}$.\\
Here, the \emph{frequency-typical sets} $T_{p,l^{-1/8}}$ are defined as $T_{p,l^{-1/8}}:=\{x^l:\|\frac{1}{l}N(\cdot|x^l)-p(\cdot)\|_1\leq l^{-1/8}\}$.\\
An application of Theorem \ref{thm:winter-existence-of-codes} then yields a sequence of codes satisfying
\begin{align}
\liminf_{l\to\infty}\frac{1}{l}\log M_l=\chi(p,\mathcal V)\qquad\mathrm{and}\qquad \forall l\in\nn:\ \ \min_{m\in[M_l]}\tr\{\sqrt{\Xi_m}\mathcal V^{\otimes l}(u^l_m)\sqrt{\Xi_m}\}\geq1-l^{-1/4},
\end{align}
where we did not write out the dependence of the operators $\Xi_m$ on $l$ (in order to spare a few indices and get a more compact notation).
The slightly unconventional way of writing above formula using the obvious decomposition $\Xi_m=\sqrt{\Xi_m}\sqrt{\Xi_m}$ is due to our intent to apply the gentle measurement operator Lemma for ensembles later. An important additional property of our code is that every codeword $u^l_m$ satisfies $\|\frac{1}{l}N(\cdot|u^l_m)-p(\cdot)\|\leq l^{-1/8}$.\\
\textbf{Construction of code for the private messages:} First, assume for the moment that a given codeword $u^l$ can be written as $u^l=\prod_{u\in\bU}u^{N(u|u^l)}$. We hereby implicitly and without loss of generality assume that an ordering of the symbols of $\bU$ is given: $i\geq j\Rightarrow u_i\geq u_j$. Let us also use the abbreviation $l_u:=N(u|u^l)$.\\
Now, with the same convention on the ordering as before, parse $\bX^l\times\bY^l$ as
\begin{align}
\bX^l\times\bY^l=\prod_{u\in\bU}(\bX\times\bY)^{l_u}.
\end{align}
On each of the above blocks $(\bX\times\bY)^{l_u}$, the work \cite{fawzi-hayden-savov-sen-wilde} describes how to perform a random choice of $K_u$ codewords $\{x^l_{k_u}\}_{k_u=1}^{K_u}\subset\bX^{l_u}$ and $T_u$ codewords $\{y^n_{t_u}\}_{t_u=1}^{T_u}$ in $\bY^{l_u}$ together with a decoding POVM $(\Lambda_{k_u,t_u}(\{x^n_{k_u}\}_{k_u=1}^{K_u},\{y^n_{t_u}\}_{t_u=1}^{T_u}))_{k_u=1,t_u=1}^{K_u,T_u}$ such that a good code for the ccq-MAC $\mathcal W$ is obtained. More precisely, the codewords are chosen at random, all independently from one another, and on each block $(\bX\times\bY)^{l_u}$ according to $r^{\otimes l_u}(\cdot|u)$ and $s^{\otimes l_u}(\cdot|u)$. Due to independence of the choice of codewords, the results of \cite{fawzi-hayden-savov-sen-wilde} (their Theorem 2 and the corresponding proof) guarantee that this can be done for numbers $(K_u,T_u)_{u\in\bU}$ satisfying for each $u$ the inequalities and some fixed $\delta>0$
\begin{align}
\frac{1}{l_u}\log(K_u)&\geq \chi(X;Q|Y,U=u)-\delta\\
\frac{1}{l_u}\log(T_u)&\geq \chi(Y;Q|X,U=u)-\delta\\
\frac{1}{l_u}[\log(K_u)+\log(L_u)]&\geq \chi(X,Y;Q|U=u)-\delta,
\end{align}
where the mutual informations are being evaluated with respect to the distributions $rs(\cdot,\cdot|u):=r(\cdot|u)s(\cdot|u)\in\mathfrak P(\bX\times\bY)$, and such that the expected error over the random choice of codewords and in uniform average over all the messages $(k_u,t_u)\in[K_u]\times[T_u]$ goes to zero for $l_u$ tending to infinity (since each $u^l\in T_{p,l^{-1/8}}$).
More precisely for each $u\in\bU$ and identifying a collection of codewords with the map $\mathcal C_u:[K_u]\times[T_u]\to\bX^{l_u}\times\bY^{l_u}$ satisfying $\mathcal C_u(k_u,t_u)=(x^{l_u}_{k_u},y^{l_u}_{t_u})\ \forall\ (k_u,t_u)\in[K_u]\times[T_u]$:
\begin{align}
\sum_{\mathcal C_u}\frac{\mathbb P(\mathcal C_u)}{K_uT_u}&\sum_{k_u,t_u=1}^{K_u,T_u}\tr\{\mathcal W^{\otimes l_u}(\mathcal C_u(k_u,t_u)\Lambda(\mathcal C_u)_{k_u,t_u}\}\\
&=\sum_{x^{l_u},y^{l_u}}r^{\otimes l_u}(x^{l_u}|u)s^{\otimes l_u}(y^{l_u}|u)\tr\{\mathcal W^{\otimes l_u}(x^{l_u},y^{l_u})A_u(x^{l_u},y^{l_u})\}\\
&\geq1-\nu(l_u),
\end{align}
where the operators $A_u(\cdot)$ are defined by
\begin{align}
A_u(x^{l_u},y^{l_u}):=\sum_{m=2}^{K_u}\sum_{n=2}^{T_u}\prod_{i,j=2}^{K_u,T_u}rs^{\otimes l_u}(x^{l_u}_i,y^{l_u}_j|u)\left(\sum_{k,t=2}^{K_u,T_u}P_{x^{l_u}_ky^{l_u}_t}\right)^{-1/2}P_{x^{l_u}y^{l_u}}\left(\sum_{k,t=2}^{K_u,T_u}P_{x^{l_u}_ky^{l_u}_t}\right)^{-1/2},
\end{align}
and $\mathbb P(\mathcal C_u)$ denotes the probability that the code $\mathcal C_u$ is chosen, and a function $\nu:\mathbb N\to\mathbb R$ satisfying $\lim_{n\to\infty}\nu(n)=0$.\\
Our requirement that $u^l\in T_{p,l^{-1/8}}$ then additionally yields (for large enough $l\in\nn$) that
\begin{align}
\frac{1}{l}\log(\prod_{u\in\bU}K_u)&\geq \chi(X;Q|Y,U)-2\delta\\
\frac{1}{l}\log(\prod_{u\in\bU}T_u)&\geq \chi(Y;Q|X,U)-2\delta\\
\frac{1}{l}[\log(\prod_{u\in\bU}K_u)+\log(T_l)]&\geq \chi(X,Y;Q|U)-2\delta\\
\frac{1}{l}\log(\prod_{u\in\bU}M_u)&\geq \chi(U;Q)-2\delta.
\end{align}
Additionally measuring $\Xi_m$ only adds another, asymptotically vanishing error due to the gentle operator lemma for ensembles. In that way, \emph{both} the message $m$ and the messages $\prod_{u\in\bU}(k_u,t_u)$ can be decoded simultaneously. We will make this more explicit at the end of the proof for general (not ordered) codewords $u^l_m$.\\
This case actually requires no additional reasoning, apart from adding a permutation on the whole system and putting an additional index $m$ everywhere to account for that. The permutation will simply re-order $u^l_m$ so that we can apply our above reasoning.\\
So, let us assume that a codeword $u^l_m$ has some arbitrary ordering of the symbols, let $\sigma_m\in S_l$ be a permutation with the property that $\sigma_m (u^n):=(u_{\sigma_m^{-1}(1)},\ldots,u_{\sigma_m^{-1}(l)})$ is well-ordered. The code we employ in that case is obtained by applying $\sigma_m$ to each codeword and the operation $\sigma_m\cdot\sigma_m^{-1}$ to each POVM element.\\
We combine these codes to get one for each codeword $u^l_m$ with lengths $l_{u}^m:=N(u|u^l_m)$. In order to save some indices we do not write the dependence of the numbers $T_u$ and $K_u$ on the lengths $l_u^m$ out at this point, we will later on take the minimum over all choices of the codeword $m$ for both of $T_u$ and $K_u$ anyway. So, for an arbitrary codeword $u^l_m$ we get:
\begin{align}
(\prod_{u\in\bU}\sum_{\mathcal C_{m,u}}\frac{\mathbb P(\mathcal C_{m,u})}{K_uT_u})&\sum_{k_u,t_u=1}^{K_u,T_u}\tr\{\mathcal W^{\otimes l_u^m}(\mathcal C_{m,u}(k_u,t_u)\Lambda(\mathcal C_{m,u})_{k_u,t_u}\}\\
&=\prod_{u\in\bU}\sum_{x^{l_u^m},y^{l_u^m}}r^{\otimes l_u^m}(x^{l_u^m}|u)s^{\otimes l_u^m}(y^{l_u^m}|u)\tr\{\mathcal W^{\otimes l_u^m}(x^{l_u^m},y^{l_u^m})A_u(x^{l_u^m},y^{l_u^m})\}\\
&\geq\prod_{u\in\bU}(1-\nu(l_u^m))\\
&\geq1-\sum_{u\in\bU}\nu(l_u^m).
\end{align}
The operators $A_u(\cdot,\cdot)$ (we refrain from ballasting them with further indices, the dependence on $m$ will be clear from the argument) are defined (see \cite{fawzi-hayden-savov-sen-wilde}) as
\begin{align}
A_u(x^{l_u^m},y^{l_u^m}):=\sum_{m,n=2}^{K_u,T_u}\prod_{i,j=2}^{K_u,T_u}rs^{\otimes l_u^m}(x^{l_u^m}_i,y^{l_u^m}_j|u)\left(\sum_{k,t=2}^{K_u,T_u}P_{x^{l_u^m}_ky^{l_u^m}_t}\right)^{-\frac{1}{2}}P_{x^{l_u^m}y^{l_u^m}}\left(\sum_{k,t=2}^{K_u,T_u}P_{x^{l_u^m}_ky^{l_u^m}_t}\right)^{-\frac{1}{2}},
\end{align}
and the symbols $P_{x^{l_u^m}y^{l_u^m}}$ appearing above denote the weakly typical subspaces for $\mathcal W$, as defined in \cite{fawzi-hayden-savov-sen-wilde}, (their POVM construction on page five). The reason for writing $A_u(\cdot,\cdot)$ in this form will become apparent soon.\\
Using for the transmission of private messages the sets $\prod_u[K_u]$ and $\prod_u[T_u]$ (which, setting $K_l:=\prod_uK_u$ and $T_l:=\prod_uT_u$ may be identified with $[K_l]$ and $[T_l]$, we arrive for each $u^l_m$ ($m\in[M_l]$) at the existence of a random choice of codes having asymptotically perfect performance with respect to the average error criterion and satisfying
\begin{align}
\frac{1}{l}\log(K_l)&\geq \chi(X;Q|Y,U)-2\delta\\
\frac{1}{l}\log(T_l)&\geq \chi(Y;Q|X,U)-2\delta\\
\frac{1}{l}[\log(K_l)+\log(T_l)]&\geq \chi(X,Y;Q|U)-2\delta\\
\frac{1}{l}\log(M_l)&\geq \chi(U;Q)-2\delta,
\end{align}
once $l$ is large enough. Note that this also implies that
\begin{align}
\frac{1}{l}\log(K_lM_lT_l)\geq\chi(X,Y;Q)-4\delta.
\end{align}
A crucial property in the derivation of above estimates is that the maps $p\mapsto I(A;B|U)$ are continuous once the systems $A$ and $B$ are finite, and that for every one of our codewords $u^l_m$ it holds $\|\frac{1}{l}N(\cdot|u^l_m)-p\|\leq l^{-1/8}$.\\
Now we are in the position to apply the gentle operator Lemma for ensembles. For each fixed $u^l_m$ let $\bar u^l_m$ be ordered according to $\bar u_i\leq\bar u_j\Leftarrow i\leq j$. The random choice of codes $\mathcal C_m$ obtained from concatenating codewords $(k,t)=\prod_{u\in\bU}(k_u,t_u)$ as $\mathcal C_m(k,t):=\sigma_m(\prod_{u\in\bU}\mathcal C_u(k_u,t_u))$ and corresponding POVMs $\Lambda(\mathcal C_m)_{k,t}:=\sigma_m(\bigotimes_{u\in\bU}\Lambda(\mathcal C_u)_{k_u,t_u})\sigma_m^{-1}$ (where $\sigma_m$ here denotes the usual action of a permutation on $\kr^{\otimes l}$ by permuting the tensor factors) such that $\mathbb P(\mathcal C_m)=\prod_u\mathbb P(\mathcal C_{m,u})$ satisfies the following:
\begin{align}
\sum_{\mathcal C_m}\frac{\mathbb P(\mathcal C_m)}{K_mT_m}\sum_{k,t=1}^{K_m,T_m}&\tr\{\sqrt{\Xi_m}\mathcal W^{\otimes l}(\mathcal C_m(k,t))\sqrt{\Xi_m}\Lambda(\mathcal C_m)_{k,t}\}\\
&=\sum_{x^{l},y^{l}}rs^{\otimes l}(x^{l},y^{l}|u^l_m)\tr\{\sqrt{\Xi_m}\mathcal W^{\otimes l}(x^{l},y^{l})\sqrt{\Xi_m}(\bigotimes_{u\in\bU}A_u(x^{l_u^m},y^{l_u^m}))\}\\
&\geq\sum_{x^{l},y^{l}}rs^{\otimes l}(x^{l},y^l|u^l)[\tr\{\mathcal W^{\otimes n}(x^{n},y^{n})(\bigotimes_{u\in\bU}A_u(x^{l_u^m},y^{l_u^m}))\}\\
&\qquad\qquad\qquad-\|\sqrt{\Xi_m}\mathcal W^{\otimes l}(x^{l},y^{l})\sqrt{\Xi_m}-\mathcal W^{\otimes l}(x^{l},y^{l})\|_1]\\
&\geq1-\sum_{u\in\bU}\nu(l_u^m)-6\sqrt{l^{-1/4}}.
\end{align}
This proves that for every $l\in\nn$ and $m\in[M_l]$ there exist $K_{m,l},\ T_{m,l}$ and corresponding codewords $(x^l_{k_m})_{k_m\in[K_{m,l}]},\ (y^l_{t_m})_{t^m\in[T_{m,l}]}$ satisfying the above cardinality bounds hold as well as POVMs $(\Lambda^{(m)}_{k_m,t_m})_{k_m,t_m=1}^{K_m,T_m}$ such that
\begin{align}
\frac{1}{K_{m,l},T_{m,l}}\sum_{k_m,t_m=1}^{K_{m,l},T_{m,l}}\tr\{\sqrt{\Xi_m}\mathcal W^{\otimes l}(x^l_{k_m},y^l_{t_m})\sqrt{\Xi_m}\Lambda^{(m)}_{k_m,t_m}\}\geq1-\sum_{u\in\bU}\nu(l_u^m)-6\sqrt{l^{-1/4}}.
\end{align}
Since throwing away some codewords never decreases the error, we may well assume that all the numbers $K_{m,l},T_{m,l}$ ($1\leq m\leq M_l$) are in fact equal to $K_l:=\min\{K_{m,l}\}_{m\in M_l}$ and $T_l:=\min\{T_{m,l}\}_{m\in M_l}$ and it then (finally) holds, with the POVM
\begin{align}
\Delta_{k,t,m}:=\sqrt{\Xi_m}\Lambda_{k_m,t_m}^{(m)}\sqrt{\Xi_m}\qquad(\ \mathrm{for\ all}\ k\in[K_l],\ t\in[T_l],\ m\in[M_l]\ ),
\end{align}
the lower bound
\begin{align}
\frac{1}{K_lT_lM_l}\sum_{m=1}^{M_l}\sum_{k,l=1}^{K_l,T_l}\tr\{\mathcal W^{\otimes l}(\mathcal C_m(k,l))\Delta_{k,l,m}\}\geq1-\min_{m\in[M_l]}\sum_{u\in\bU}\nu(l_u^m)-6\sqrt{l^{-1/4}}
\end{align}
on the probability of successful transmission of our messages. Since the right hand side in above inequality goes to zero for $l$ going to infinity due to our choice $u^l_m\in T_{p,l^{-1/8}}\ \forall\ m\in[M_l],\ \forall\ l\in\nn$, we are done.
\end{proof}
\end{subsection}
\begin{subsection}{\label{subsec:direct-part-for-MAC-with-conferencing-encoders}Code construction for the MAC with conferencing encoders}
The construction of codes for the MAC with conferencing encoders (or, to be more precise, the proof of Theorem \ref{theorem:direct-part-for-MAC-with-conferencing-encoders}) from the existing ones for the MAC with common messages is now straightforward along the lines of \cite{willems}, and we will resist the temptation of producing a few redundant pages at this point.\\
Rather, we will give an informal scetch of proof:\\
Consider the two senders with conferencing capacities $C,D$ attempting to send messages at rates $R_M,R_N$. Define the numbers $c:=\min\{R_M,C\}$ and $d:=\min\{R_N,D\}$, and make a disjoint partitioning of the message set $[2^{nR_M}]=\cup_{i=1}^c M_i$ into subsets all having the same size, and the same for the other sender: $[2^{nR_N}]=\cup_{i=1}^d N_i$. Of course, a partitioning into sets of the exact same size is not always possible and there will usually be one set left that has a different size than all the others. Asymptotically however, these leftover sets play only a negligible role.\\
The senders now send as a conferencing message the index of the partition that their message is chosen from, and the conferencing only uses this one step.\\
The pairs $(i,j)$ of indices numbering the partitions can then be considered common messages of the two senders, and the code for the ccq-MAC with common messages from Theorem \ref{theorem:direct-part-of-MAC-with-common-message} is used. The requirement that all the sets $N_i,M_i$ are of the same size ensures that $(i,j)$ is evenly distributed, and this is true with a small and asymptotically vanishing error.\\
More details can for example be picked up in the original paper \cite{willems} by Willems.
\end{subsection}
\begin{subsection}{\label{subsec:converse-for-MAC-with-common-message}Converse for the MAC with a common message}
This proof follows in many ways the same reasoning as the one of the converse for the MAC with conferencing encoders. It should be noted that the conditional independence of the individual messages given the conferencing result (equations (\ref{eqn:willems-20}) to (\ref{eqn:willems-22})) is present in the scenario with common messages as well: Both senders act independently from one another if they want to, but can as well condition their encoding on the joint message. Thus in general, their individual messages are only independent given the joint message. We will give a rigorous proof.
\begin{proof}[Proof of Theorem \ref{theorem:converse-for-MAC-with-common-message}]
We will use the following Lemma, which is a straightforward generalization of Lemma 1 in \cite{slepian-wolf-MAC}, and therefore presented without proof:
\begin{lemma} Let $M,K,L$ be independent random variables with values in the finite sets $\bM,\bK,\bL$, each distributed evenly on the respective set. Let $\mathcal V\in CQ(\bX,\bY,\kr)$ and encoding functions $a:\bM\times\bK\to\bX$, $b:\bM\times\bK\to\bY$ be given, as well as a POVM $\mathbf D\in\mathcal M_{|\bM\times\bK\times\bL|}(\kr)$. Define the distribution $p\in\mathfrak P(\bM\times\bK\times\bL\times\bM'\times\bK'\times\bL')$ (where $\bM=\bM'$ and so on)
\begin{align}
p(m,k,l,m',k',l'):=\frac{1}{|\bM\times\bK\times\bL|}\tr\{\mathcal V(a(m,k),b(m,l))D_{k',m',l'}\},
\end{align}
and the quantity
\begin{align}
p_e:=1-\sum_{k,l,m}p(m,k,l,m,k,l).
\end{align}
Then for $p_e\leq1/2$,
\begin{align}
H(K|(M',K',L'),M,L)\leq p_e\log|\bK|+1,\\
H(L|(M',K',L'),M,K)\leq p_e\log|\bL|+1,\\
H(K,L|(M',K',L'),M)\leq p_e\log|\bK\times\bL|+1,\\
H(M,K,L|M',K',L')\leq p_e\log|\bK\times\bL\times\bM|+1.
\end{align}
\end{lemma}
Let a sequence $(\mathfrak C_l)_{l\in\nn}$ of codes for the MAC $\mathcal W$ with common messages be given such that $\eps_l:=1-p_{\mathrm{s}}(\mathfrak C_l)$ satisfies $\eps_l\searrow0$ at rates $R_X,R_Y,R_C$. Defining $\delta_l:=\eps_l\cdot\log|\bK\times\bL\times\bM|+1$, we see that (in the same manner as in Subsection \ref{subsec:proof-of-converse-for-conferencing-MAC}), we get the inequalities
\begin{align}
\log(K_l)&\leq I(K^l;\hat K^l,\hat M^l,\hat L^l|L^l,M^l)+\delta_l\\
\log(L_l)&\leq I(L^l;\hat K^l,\hat M^l,\hat L^l|K^l,M^l)+\delta_l\\
\log(K_l\cdot L_l)&\leq I(K^l,L^l;\hat K^l,\hat M^l,\hat L^l|M^l)+\delta_l\\
\log(K_l\cdot L_l\cdot M_l)&\leq I(K^l,L^l,M^l;\hat K^l,\hat M^l,\hat L^l)+\delta_l.
\end{align}
It then follows, from the structure of the encoding, the Holevo bound and Lemma \ref{lemma:winter-subadditivity}, applied in that order:
\begin{align}
I(K^l;\hat K^l,\hat M^l,\hat L^l|L^l,M^l)&\leq I(X^l;\hat K^l,\hat M^l,\hat L^l|Y^l,M^l)\\
&\leq \chi(X^l;Q^l|Y^l,M^l)\\
&\leq\sum_{i=1}^lI(X_i;Q_i|Y_i,M^l).
\end{align}
The same argument leads to the upper bound
\begin{align}
I(L^l;\hat K^l,\hat M^l,\hat L^l|K^l,M^l)&\leq\sum_{i=1}^l\chi(Y_i;Q_i|X_i,M^l).
\end{align}
The data processing inequality together with the Holevo bound and Lemma \ref{lemma:winter-subadditivity} in its conditional form (see Remark \ref{remark:winter-subadditivity}) yields
\begin{align}
I(K^l,L^l;\hat K^l,\hat M^l,\hat L^l|M^l)&\leq I(X^l,Y^l;\hat K^l,\hat M^l,\hat L^l|M^l)\\
&\leq\sum_{i=1}^l I(X_i,Y_i;Q_i|M^l).
\end{align}
We are finally left with the four inequalities
\begin{align}
\log(K_l)&\leq \sum_{i=1}^l\chi(X_i;Q_i|Y_i,M^l)+\delta_l\\
\log(L_l)&\leq \sum_{i=1}^l\chi(Y_i;Q_i|X_i,M^l)+\delta_l\\
\log(K_l\cdot L_l)&\leq\sum_{i=1}^l\chi(X_i,Y_i;Q_i|M^l)+\delta_l\\
\log(K_l\cdot L_l\cdot M_l)&\leq \sum_{i=1}^l\chi(X_i,Y_i;Q_i)+\delta_l.
\end{align}
(the last one being an obvious consequence of the converse for the single sender cq-channel). As in the proof for the MAC with conferencing encoders, this implies that for every $l\in\nn$ the points
\begin{align*}
P_{1,l}:=\sum_{i=1}^l\chi(X_i;Q_i|Y_i,U^l), \ P_{2,l}:=\sum_{i=1}^l\chi(Y_i;Q_i|Y_i,U^l),\\
P_{3,l}:=\sum_{i=1}^l\chi(X_i,Y_i;Q_i|U^l), \ P_{4,l}:=\sum_{i=1}^l\chi(X_i,Y_i;Q_i).
\end{align*}
are contained in the four dimensional closed (and also convex, due to the freedom in the choice of the alphabets $\bU$) region $\mathfrak R^4$ defined again by
\begin{align*}
\mathfrak R^4:=\left\{\right.(R_i)_{i=1}^4|&R_1\leq \chi(X;Q|Y,U),\ R_2\leq \chi(Y;Q|X,U),\ R_3\leq \chi(X,Y;Q|U),\ R_4\leq \chi(X,Y;Q)\\
&\mathrm{for\ some\ alphabet\ }\bU\ \mathrm{and\ }(U,X,Y,Q)\ \mathrm{a\ cq\ system\ as\ in\ Theorem\ \ref{theorem:direct-part-of-MAC-with-common-message}}\left\}\right.
\end{align*}
But then for all $\eps>0$ and $l\geq L=L(\eps)$ we have that
\begin{align}
S_X\leq P_{1,l}+3\eps,\qquad S_Y\leq P_{2,l}+3\eps,\qquad S_X+S_Y\leq P_{3,l}+3\eps,\qquad S_X+S_Y+S_C\leq P_{4,l}+3\eps
\end{align}
and whence for all $\eps>0$ the vector $(S_X,S_Y,S_X+S_Y,S_X+S_Y+S_C)$ is contained in $\mathfrak R^4+B_{3\eps}$, where again for two sets $A,B\subset \mathbb R^4$ their sum is defined by $A+B:=\{a+b:a\in A,\ b\in B\}$ and for an $\eps>0$ we the ball of radius $\eps$ (in one-norm) around the point $0\in\mathbb R^4$ is defined by $B_\eps:=\{x\in\mathbb R^4:\sum_i|x_i|\leq\eps\}$.\\
Therefore $(S_X,S_Y,S_X+S_Y,S_X+S_Y+S_C)\in\mathrm{cl}(\mathfrak R^4)$, and from that we deduce that $(S_X,S_Y,S_C)\in\mathrm{cl}(\cup_p\mathfrak R_{p,\mathrm{comm}}(\mathcal W))=\mathfrak R_{\mathrm{comm}}(\mathcal W)$, which completes the proof of the converse for the MAC with conferencing encoders.\\
It remains to see that the constraint on the size of the helping alphabet $\bU$ is valid. But the argument leading to that bound is exactly the same as in the proof of the converse for the MAC with conferencing encoders, so we refer the reader to that.
\end{proof}
\end{subsection}
\end{section}
\emph{Acknowledgements.}
This work was supported by the DFG via grant BO 1734/20-1 (H.B.) and by the BMBF via grant 01BQ1050 (H.B., J.N.).


\end{document}